\documentclass[twoside,leqno,twocolumn,10pt]{article}  
\usepackage{ltexpprt} 

\usepackage{enumitem}
\setenumerate[1]{itemsep=0pt,partopsep=0pt,parsep=\parskip,topsep=3pt}
\setitemize[1]{itemsep=0pt,partopsep=0pt,parsep=\parskip,topsep=3pt}
\setdescription{itemsep=0pt,partopsep=0pt,parsep=\parskip,topsep=3pt}

\usepackage{tabularx}
\newcolumntype{L}[1]{>{\raggedright\arraybackslash}p{#1}}
\newcolumntype{C}[1]{>{\centering\arraybackslash}p{#1}}
\newcolumntype{R}[1]{>{\raggedleft\arraybackslash}p{#1}}

\usepackage{caption}
\usepackage{hyperref}
\usepackage{algorithm}
\usepackage{algorithmic}
\usepackage{graphicx}
\usepackage{subcaption}
\usepackage{verbatim}
\usepackage{amsmath}
\usepackage{amssymb}
\usepackage{multirow}

\begin{document}

\title{\Large On Spectral Graph Embedding: A Non-Backtracking Perspective and Graph Approximation}
\author{
Fei Jiang\thanks{Department of Computer Science, Peking University. allen.feijiang@gmail.com, jxu@pku.edu.cn}
\and 
Lifang He\thanks{Department of Healthcare Policy and Research, Weill Cornell Medical School, Cornell University. lifanghescut@gmail.com}
\and
Yi Zheng\footnotemark[1]
\and
Enqiang Zhu\thanks{School of Computer Science and Educational Software, Guangzhou University. zhuenqiang@gzhu.edu.cn}
\and
Jin Xu\footnotemark[1]
\and 
Philip S. Yu\thanks{Shanghai Institute for Advanced Communication and Data Science, Fudan University, Shanghai, China}
\thanks{Department of Computer Science, University of Illinois at Chicago. psyu@uic.edu}
}
\date{}
\maketitle


\begin{abstract} \small\baselineskip=9pt 
Graph embedding has been proven to be efficient and effective in facilitating graph analysis. In this paper, we present a novel spectral framework called NOn-Backtracking Embedding (NOBE), which offers a new perspective that organizes graph data at a deep level by tracking the flow traversing on the edges with backtracking prohibited. Further, by analyzing the non-backtracking process, a technique called graph approximation is devised, which provides a channel to transform the spectral decomposition on an edge-to-edge matrix to that on a node-to-node matrix. Theoretical guarantees are provided by bounding the difference between the corresponding eigenvalues of the original graph and its graph approximation. Extensive experiments conducted on various real-world networks demonstrate the efficacy of our methods on both macroscopic and microscopic levels, including clustering and structural hole spanner detection.
\end{abstract}
\section{Introduction}
Graph representations, which describe and store entities in a node-interrelated way \cite{shen2010} (such as adjacency matrix, Laplacian matrix, incident matrix, etc), provide abundant information for the great opportunity of mining the hidden patterns. However, this approach poses two principal challenges: 1) one can hardly apply off-the-shelf machine learning algorithms designed for general data with vector representations, and adapt them to the graph representations and 2) it's intractable for large graphs due to limited space and time constraints. Graph embedding can address these challenges by representing nodes using meaningful low-dimensional latent vectors. 

Due to its capability for assisting network analysis, graph embedding has attracted researchers' attention in recent years \cite{tang2015line, perozzi2014deepwalk, grover2016node2vec, wang2016structural}. The goal of a good graph embedding algorithm should be preserving both macroscopic structures (e.g., community structure) and microscopic structures (e.g., structural hole spanner) simultaneously. However, an artless graph embedding algorithm will lead to unsatisfactory low-dimensional embeddings in which meaningful information may lose or be indistinguishable. For example, the pioneering work \cite{tang2015line} mainly focusing on locally preserving the pairwise distance between nodes can result in the missing of dissimilarity. As a result, it may fail to preserve the community membership, as shown later in the experimental results in Table \ref{tab:modularity_community}. Some works \cite{grover2016node2vec, perozzi2014deepwalk} attempt to preserve high-order proximity between nodes by considering truncated random walk. Since conventional truncated random walk is a Markov chain without any examining the special structure of networks, key nodes information (e.g., structural hole spanners, outliers) will be unrecognizable. As far as we know, present approaches cannot achieve the graph embedding goal well. 
\begin{figure}[!tb]
\centering
\begin{subfigure}{0.44\linewidth}
\centering
\includegraphics[width=\linewidth]{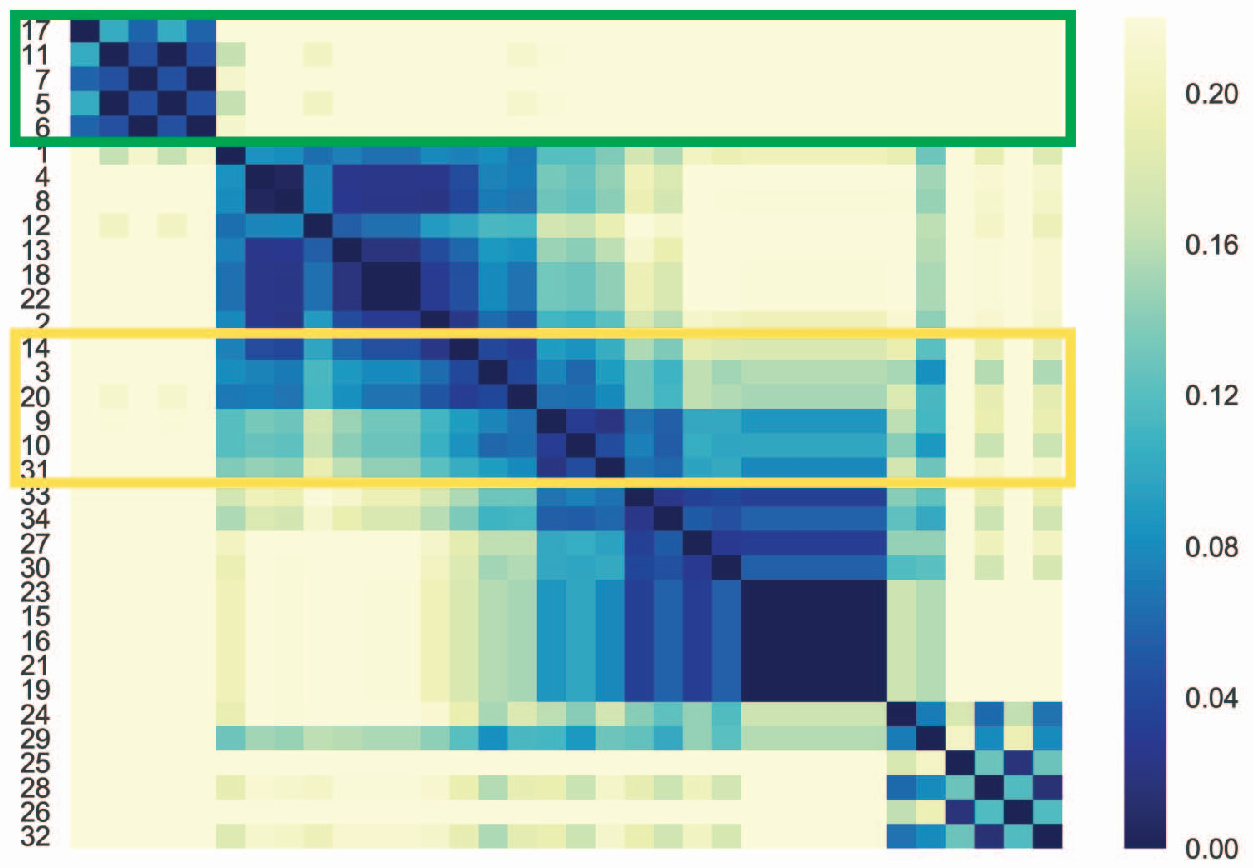}
\caption{}
\label{fig:map_example}
\end{subfigure}
\centering
\begin{subfigure}{0.44\linewidth}
\centering
\includegraphics[width=\linewidth]{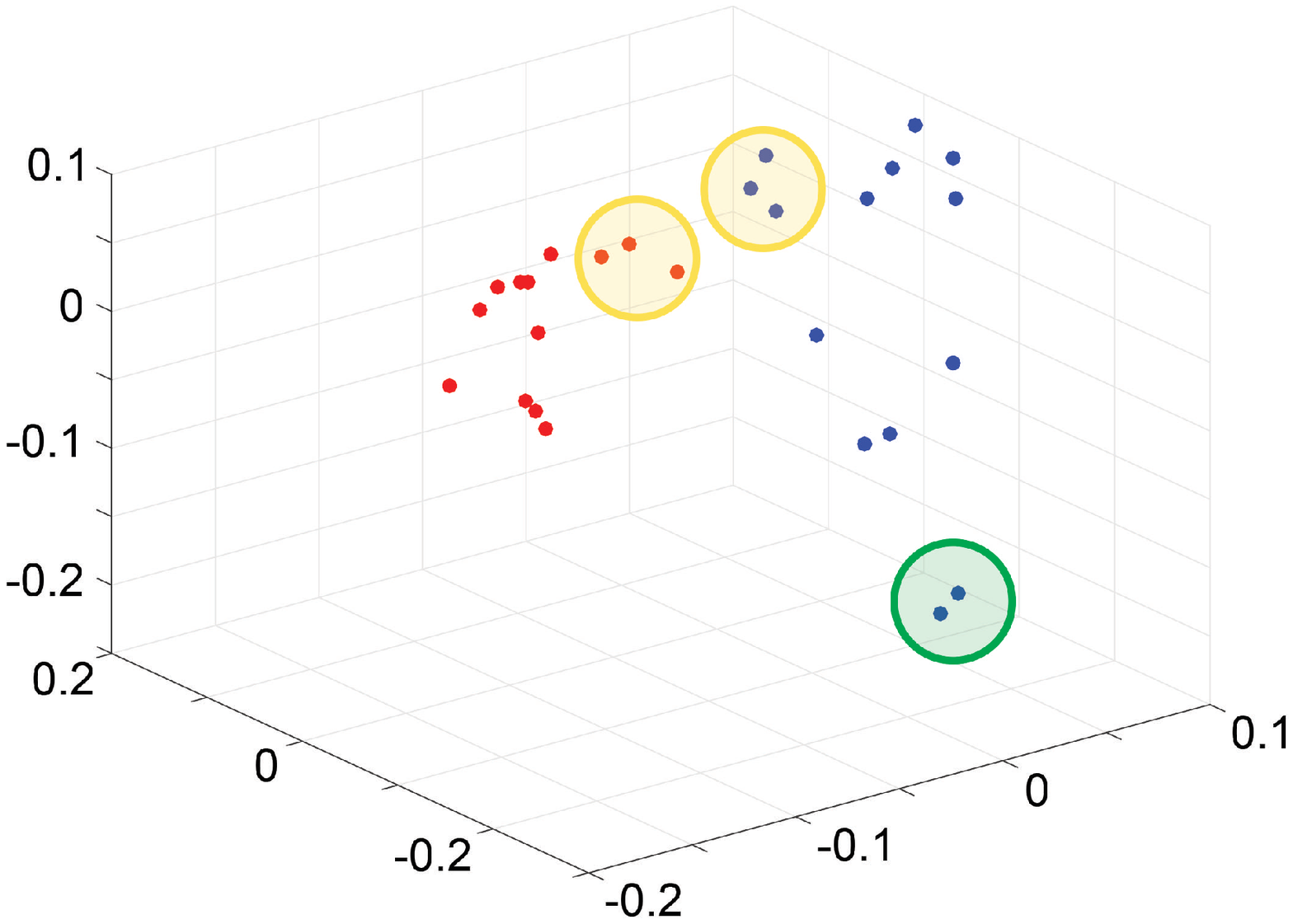}
\caption{}
\label{fig:point_example}
\vspace{-5pt}
\end{subfigure}
\caption{An embedding visualization of our method NOBE on \emph{karate} network. (a) The distance between the embedding vectors is plotted, where x and y axes represent node ID respectively. Community structure, structural holes (yellow frame) and outliers (green frame) are easily identified; (b) Ground truth communities are rendered in different colors, which is well preserved in the embedding subspace. Again, structural holes and outliers are marked with yellow and green circles respectively.}
\label{fig:karate}
\end{figure}

In this paper, to fulfill the goal of preserving macroscopic and microscopic structures, we propose a novel graph embedding framework NOBE and its graph approximation algorithm NOBE-GA. The main contributions of this paper are summarized as follows:
\begin{itemize}
\setlength{\itemsep}{5pt}
\setlength{\parsep}{5pt}
\setlength{\parskip}{5pt}

\item We develop an elegant framework NOn-Backtracking Embedding (NOBE), which jointly exploits a non-backtracking random walk and spectral graph embedding technique (in Section \ref{sec:Method}). The benefits of NOBE are: 1) From an edge perspective, we encode the graph data structure to an \emph{oriented line graph} by mapping each edge to a node, which facilitates to track the flow traversing on edges with backtracking prohibited. 2) By figuring out node embedding from the \emph{oriented line graph} instead of the original graph, community structure, structural holes and outliers can be well distinguished (as shown in Figures \ref{fig:karate} and \ref{fig:intuition}).




\item Graph approximation technique NOBE-GA is devised by analyzing the pattern of non-backtracking random walk and switching the order of spectral decomposition and summation (in Section \ref{sec:Approximations}). It reduces the complexity of NOBE with theoretical guarantee.
Specifically, by applying this technique, we found that conventional spectral method is just a reduced version of our NOBE method.

\item In section \ref{sec:Experiments}, we also design a metric RDS based on embedding community structure to evaluate the nodes' topological importance in connecting communities, which facilitates the discovery of structural hole (SH) spanners. Extensive experiments conducted on various networks demonstrate the efficacy of our methods in both macroscopic and microscopic tasks.
\end{itemize}

\section{Related Work}
\label{sec:related work}
Our work is mainly related to graph embedding and non-backtracking random walk. We briefly discuss them in this section.
\subsection{Graph Embedding}
Several approaches aim at preserving first-order and second-order proximity in nodes' neighborhood. \cite{wang2016structural} attempts to optimize it using semi-supervised deep model and \cite{tang2015line} focuses on large-scale graphs by introducing the edge-sampling strategy. To further preserve global structure of the graph, \cite{qiu2007clustering} explores the spectrum of the commute time matrix and  \cite{perozzi2014deepwalk} treats the truncated random walk with deep learning technique. Spectral method and singular value decomposition are also applied to directed graphs by exploring the directed Laplacian matrix  \cite{chen2007directed} or by finding the general form of different proximity measurements \cite{ou2016asymmetric}. Several works also consider joint embedding of both node and edge representations \cite{Xu2017,abu2017learning} to give more detailed results. By contrast, our work can address graph embedding using a more expressive and comprehensive spectral method, which gives more accurate vector representations  in a more explainable way yet with provable theoretical guarantees.

\subsection{Non-backtracking Random Walk}
Non-backtracking strategy is closely related to Ihara's zeta function, which plays a central role in several graph-theoretic theorems \cite{kempton2015high,bordenave2015non}. Recently, in machine learning and data mining fields, some important works have been focusing on developing non-backtracking walk theory. \cite{krzakala2013spectral} and \cite{saade2014spectral} demonstrate the efficacy of the spectrum of non-backtacking operator in detecting communities, which overcomes the theoretic limit of classic spectral clustering algorithm and is robust to sparse networks. \cite{morone2015influence} utilizes non-backtracking strategy in influence maximization, and the nice property of locally tree-like graph is fully exploited to complete the optimality proof. The study of eigenvalues of non-backtracking matrix of random graphs in \cite{bordenave2015non} further confirms the spectral redemption conjecture proposed in \cite{krzakala2013spectral} that above the feasibility threshold, community structure of the graph generated from stochastic block model can be accurately discovered using the leading eigenvectors of non-backtracking operator. However, to the best of our knowledge, there is no work done on analyzing the theory of non-backtracking random walk for graph embedding purposes.
\section{Methodology}
\label{sec:Method}
In this section, we firstly define the problem. Then, our NOBE framework is given in detail. At last, we present graph approximation technique, followed by a discussion. To facilitate the distinction, scalars are denoted by lowercase letters (e.g., $\lambda$), vectors by bold lowercase letters (e.g., $\mathbf{y}$,$\boldsymbol{\phi}$), matrices by bold uppercase letters (e.g., $\mathbf{W}$,$\mathbf{\Phi}$) and graphs by calligraphic letters (e.g., $\mathcal{G}$). The basic symbols used in this paper are also described in Table \ref{tab:Symbol}. 
\begin{table}[htbp]
\scriptsize
\caption{\bf List of basic symbols}
\centering 
\begin{tabular}{c|l}
\hline
Symbol&Definition\\
\hline
$\mathcal{G}$ & Original graph (edge set omitted as $\mathcal{G}=(V,\boldsymbol{\cdot},\mathbf{W})$)\\

$V,E,n,m$& Node, edge set and its corresponding volume in $\mathcal{G}$\\

$d_\mathcal{G}(v)$& Degree of node $v$ (without ambiguity denoted as $d(v)$)\\

$N(v)$ & the neighbor set of node $v$ in $\mathcal{G}$ \\

$\mathbf{A},\mathbf{W},\mathbf{D}$& Adjacency, weighted adjacency, diagonal degree matrices\\
\hline
$\mathcal{H}$ & Oriented line graph\\

$\mathbf{P}$& Non-backtracking transion matrix\\

$ \mathbf{\mathcal{L}}$ & Directed Laplacian matrix\\

$\boldsymbol{\phi}$& Perron vector\\

$\boldsymbol{\Phi}$& Diagonal matrix with entries $\Phi(v,v)=\phi(v)$ \\

\hline
\end{tabular}
\label{tab:Symbol}
\end{table}
\vspace{-10pt}
\subsection{Problem Formulation}
\label{sec:Problem_formulation}
Generally, a graph is represented as $\mathcal{G}=(V,E,\mathbf{W})$, where $V$ is set of nodes and $E$ is set of edges ($n=|V|$, $m=|E|$). When the weighted adjacency matrix  $\mathbf{W}$ (representing the strength of connections between nodes) is presented, edge set $E$ can be omitted as  $\mathcal{G}=(V,\boldsymbol{\cdot},\mathbf{W})$. Note that when $\mathcal{G}$ is undirected and unweighted, we use $\mathbf{A}$ instead of $\mathbf{W}$. Since most machine learning algorithms can not conduct on this matrix effectively, our goal is to learn low-dimensional vectors which can be fitted to them. Specifically, we focus on graph embedding to learning low-dimensional vectors, and simultaneously achieve two objectives: decoupling nodes' relations and dimension reduction. Graph embedding problem is formulated as:


\begin{Definition}
\label{def:embedding}
(\textbf{Graph Embedding}) Given a graph $\mathcal{G}=(V,E,\mathbf{W})$, for a fixed embedding dimension $k\ll n$, the purpose of graph embedding is to learn a mapping function $f(i|\mathbf{W}): i \rightarrow \mathbf{y}_i\in \mathbb{R}^k$, for $\forall i\in V$.
\end{Definition}

\subsection{Non-Backtracking Graph Embedding}
\label{sec:sub-nobe}
We proceed to present NOBE. Inspired by the idea of analyzing flow dynamics on edges, we first embed the graph into an intermediate space from a non-backtracking edge perspective. Then, summation over the embedding on edges is performed in the intermediate space to generate accurate node embeddings. In the following, we only elaborate the detail of embedding undirected unweighted graphs, while the case for weighted graphs is followed accordingly. We first define the concept of a non-backtracking transition matrix, which specifies the probabilities that the edges directed from one node to another with backtracking prohibited.

\begin{Definition}
(\textbf{Non-Backtracking Transition Matrix}) Given an undirected unweighted graph $\mathcal{G}=(V,E,\mathbf{A})$, we define its non-backtracking transition matrix $\mathbf{P}$ as a $2m \times 2m$ matrix, which can be regarded as a random walk on directed edges of graph $\mathcal{G}$ with backtracking prohibited. Mathematically,
\label{de:nbtm}
\begin{equation}
\label{eq:nbtm}
\mathbf{P}_{[(u\rightarrow v),(x\rightarrow y)]}=
\left\{
\begin{aligned}
&\frac{1}{d_\mathcal{G}(v)-1},   \quad \text{if $v = x$  and  $ u \neq y$}.\\
&0, \quad \quad \quad \quad  \quad otherwise.
\end{aligned}
\right.
\end{equation}
where $u$, $v$, $x$, $y$ $\in V$ and $(u\rightarrow v)$, $(x\rightarrow y)$ are edges with directions taken into consideration.
\end{Definition}
By encoding a graph into a non-backtracking transition matrix, it allows the diffusion dynamics to be considered. Notice that, different from the original non-backtracking operator \cite{krzakala2013spectral}, we also take the diffusion probability of the edges into account by the definition of $\mathbf{P}$. In this way, it enables us to capture more information for complex topological structure of the graph. 

Further, non-backtracking random walk is a non-Markovian chain, which uses non-backtracking transition matrix as its transition probability matrix. To make the analysis more tractable, we transform the non-Markovian process to a Markovian process by introducing an \emph{oriented line graph}.

\begin{Definition}
\label{de:olg}
(\textbf{Oriented Line Graph}) Given an undirected unweighted graph $\mathcal{G}=(V,E,\mathbf{A})$, its oriented line graph $\mathcal{H}=(\vec{E},\boldsymbol{\cdot},\mathbf{P})$ is a directed weighted graph, whose node set is the set of oriented edges in $\mathcal{G}$, and weighted adjacency matrix is the non-backtracking transition matrix.

\end{Definition}

Figure \ref{fig:intuition} illustrates the intuition behind the oriented line graph. It can be seen that the oriented line graph has the potential ability to characterize community boundary and emphasize structural hole spanners. An intuitive graph embedding approach is to perform spectral decomposition on non-backtracking transition matrix $\mathbf{P}$. However, $\mathbf{P}$ is an asymmetric matrix, so it is not guaranteed to have real eigenvalues and eigenvectors. Also, from the definition of $\mathbf{P}$, some terms in $\mathbf{P}$ is invalid at $d_\mathcal{G}(v)=1$, for $v\in V$. We propose the Proposition \ref{pro:md2} to make full use of $\mathbf{P}$ in spectral graph embedding.

\begin{figure}[!tb]
\centering
\includegraphics[width=\linewidth]{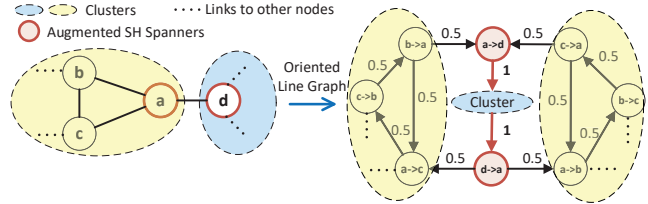}
\caption{{\bf An illustration of the intuition behind Oriented Line Graph}. Four nodes in two clusters are shown in the original graph, and edge weights, i.e., transition probabilities, are shown in the oriented line graph. If a walk from the yellow cluster randomly chooses to go to node $(a\rightarrow d)$, it must be followed by a series of walks inside the blue cluster, since backtracking is prohibited (indicated by the edge weight of $1$), and vice versa. Moreover, structural hole spanners become more evident, as node $(a\rightarrow d)$ and node $(d\rightarrow a)$ are put into crucial positions with more concentrated edges weights.}
\label{fig:intuition}
\end{figure}
\begin{proposition}
\label{pro:md2}
If the minimum degree of the connected graph $\mathcal{G}$ is at least 2, then the oriented line graph $\mathcal{H}$ is valid and strongly connected.
\end{proposition}
\begin{proof}
The proof is given in the Appendix A.1. 
\end{proof}
Proposition \ref{pro:md2} also means that under this condition the non-backtracking transition matrix $\mathbf{P}$ is irreducible and aperiodic. In particular, according to the Perron-Frobenius Theorem \cite{horn2012matrix}, it implies that for a strongly connected oriented line graph with non-negative weights, matrix $\mathbf{P}$ has a unique left eigenvector $\boldsymbol{\phi}$ with all entries positive. Let us denote $r$ as the largest real eigenvalue of matrix $\mathbf{P}$. Then,
\begin{equation*}
\boldsymbol{\phi} ^T \mathbf{P} =r \boldsymbol{\phi} ^T.
\end{equation*}
For directed weighted graphs, nodes' importance in topological structure is not determined by the degree of nodes as in an undirected graph, since directed edges coming in or going out of a node may be blocked or be rerouted back immediately in the next path. Hence, as discussed in lots of literatures \cite{chung2005laplacians}, we use Perron vector $\boldsymbol{\phi} $ to denote node importance in the oriented line graph. Our objective is to ensure that linked nodes in oriented line graph should be embedded into a close location in the embedding space. 

Suppose we want to embed nodes in the oriented line graph into one dimensional vector $\mathbf{y}$. Regarding to each edge $(e_1,e_2)$ in the oriented line graph $\mathcal{H}$, by considering its weights, our goal is to minimize $(y(e_1)-y(e_2))^2 P(e_1,e_2)$. Taking source nodes' importance indicated by $\boldsymbol{\phi}$ into consideration and summing the loss over all edges, we define our loss function as
\begin{equation}
\label{eq:loss-1}
\min_{\mathbf{y}}\sum\limits_{(e_1,e_2)\in E(\mathcal{H})}\phi(e_1)[(y(e_1)-y(e_2))^2 P(e_1,e_2)].
\end{equation}
Specifically, the Eq. (\ref{eq:loss-1}) can be written in a matrix form by the following proposition.
\begin{proposition}
\label{pro:loss-2}
Eq (\ref{eq:loss-1}) has the following form
\begin{equation}
\label{eq:loss-2}
\min_{\mathbf{y}} \mathbf{y}^T\mathcal{L}\mathbf{y},
\end{equation}
\end{proposition}
where $\mathcal{L}=\boldsymbol{\Phi} - \frac{\boldsymbol{\Phi} \mathbf{P}+\mathbf{P}^T\boldsymbol{\Phi}}{2}
$ is called \emph{combinatorial Laplacian} for directed graphs, and $\boldsymbol{\Phi}$ is a diagonal matrix with $\Phi(u,u)=\phi(u)$. 
\begin{proof}
The proof is given in the Appendix A.2.
\end{proof}

Following the idea of  \cite{chung2005laplacians}, we consider the Rayleigh quotient for directed graphs as follows:
$$
R(y)=\frac{\mathbf{y}^T\mathbf{\mathcal{L}}\mathbf{y}}{\mathbf{y}^T\boldsymbol{\Phi }\mathbf{y}}.
$$
The denominator of Rayleigh quotient takes the amount of weight distribution in the directed graph indicated by $\boldsymbol{\Phi}$ into account. Therefore, we add $\mathbf{y}^T\boldsymbol{ \Phi}\mathbf{y}=1$ as a constraint, which can eliminate the arbitrary rescaling caused by $\mathbf{y}$ and $\boldsymbol{\Phi}$. By solving Eq. (\ref{eq:loss-2}) with this constraint, we get the following eigenvector problem:
\begin{equation}
\label{eq:eigen_phi_L}
(\boldsymbol{\Phi}^{-1}\mathbf{\mathcal{L}})\mathbf{y}=\lambda \mathbf{y}.
\end{equation}
It is now clear that our task of this stage becomes selecting smallest eigenvectors of $\boldsymbol{\Phi^{-1}}\mathbf{\mathcal{L}}$ to form a vector representation for directed edges from non-backtracking perspective. By using the following proposition, we can further reduce the $\boldsymbol{\Phi^{-1}}\mathbf{\mathcal{L}}$ matrix into a more concise and elegant form.
\begin{proposition}
\label{pro:doublesum}
Both the sums of rows and columns of the non-backtracking transition matrix $\mathbf{P}$ equal one. That is, $\textbf{1}^T\mathbf{P}=\textbf{1}^T$, where $\textbf{1}$ is a column vector of ones.
\end{proposition}
\begin{proof}
The proof is given in the Appendix A.3.
\end{proof}
From the Proposition \ref{pro:doublesum}, we know that $\textbf{1}$ is a Perron vector of $\mathbf{P}$. By normalizing $\boldsymbol{\phi}$ (subject to $\sum_{e}\phi(e) =1$), we can further have $\phi(e)=\frac{1}{2m}$ for each node $e\in\vec{E}$. Then, we have
\begin{equation}
\label{eq:phi_l_tildeL}
\boldsymbol{\Phi}^{-1}\mathbf{\mathcal{L}}=\widetilde{\mathcal{L}},
\end{equation}
where $\widetilde{\mathcal{L}}=\mathbf{I}-\frac{\mathbf{P}+\mathbf{P}^T}{2}$, and can be thought of as a normalized Laplacian matrix for oriented line graphs, compared to traditional Laplacian matrix for undirected graphs. $\frac{\mathbf{P}+\mathbf{P}^T}{2}$ can be regarded as a symmetrication process on a $2m \times 2m$ matrix. Specifically, this process will be equivalent to neutralize the weights between zero and $P_{(e_1,e_2)}$ if $P_{(e_1,e_2)}$ is not null, since $P_{(e_1,e_2)}$ and $P_{(e_2,e_1)}$ cannot be nonzero at the same time. 

According to Eq \ref{eq:eigen_phi_L} and Eq \ref{eq:phi_l_tildeL}, we could obtain $k$-dimensional embedding vectors of directed edges by computing $k$ smallest non-trivial eigenvectors of $\widetilde{\mathbf{\mathcal{L}}}$. 
By summing these results over related embedding vectors, we can obtain node embeddings of graph $\mathcal{G}$. Here we introduce two sum rules: \emph{in-sum} and \emph{out-sum}. Suppose we have got a one-dimensional vector of the embedding of edges denoted by $\mathbf{g}$. For any node $u$, we define the rule of \emph{in-sum} by $g^{in}_u=\sum_{v\in N(u)}g_{v\rightarrow u}$, which sums of all the incoming edges' embeddings associated with $u$. We define the rule of \emph{out-sum} by $g^{out}_u=\sum_{v\in N(u)}g_{u\rightarrow v}$, which sums of all the outgoing edges' embeddings associated with $u$. Our graph embedding algorithm is described in algorithm \ref{alg:nobe}. 
\begin{algorithm}[!tb]
\caption{\textbf{NOBE}: NOn-Backtracking Embedding }
\label{alg:nobe}
\begin{algorithmic}[1] 
\REQUIRE ~~Graph $\mathcal{G}=(V,E)$; Embedding dimension $k$
\ENSURE  Set of embedding vectors $(\mathbf{y}_1,\mathbf{y}_2,\cdots,\mathbf{y}_n)$

\STATE Preprocess original graph $\mathcal{G}$ to meet the requirement of Proposition \ref{pro:md2}
\STATE Initialize the non-backtracking transition matrix $\mathbf{P}$ by definition \ref{de:nbtm}
\STATE Compute the second to the $k+1$ smallest eigenvectors of matrix $\widetilde{\mathbf{\mathcal{L}}}=\mathbf{I}-\frac{\mathbf{P}+\mathbf{P}^T}{2}$, denoted by $\mathbf{g}[1],\mathbf{g}[2],\cdots,\mathbf{g}[k]$
\STATE For every $u\in V$ at every dimension $i\in [1,\cdots,k]$, i.e., $y_u(i)=g[i]^{in}_u=\sum_{v\in N(u)}g[i]_{v\rightarrow u}$, based on the \emph{in-sum} rule.

\RETURN$(\mathbf{y}_1,\mathbf{y}_2,\cdots,\mathbf{y}_n)$
\end{algorithmic}
\end{algorithm} 
\normalsize
\subsection{Graph Approximation}
\label{sec:Approximations}
In the previous part, we present a spectral graph embedding algorithm NOBE, which can preserve both macroscopic and microscopic structures of the original graph. The main procedure of NOBE uses a two-step operation sequentially: eigenvector decomposition and summation of incoming edge embeddings. The first step is conducted on a $2m \times 2m$ matrix $\widetilde{\mathcal{L}}$, which is equivalent to compute the several largest eigenvectors on $\frac{\mathbf{P}+\mathbf{P}^T}{2}$, denoted as $\mathbf{\overline{P}}$. In this section, we will show how to speedup the algorithm by reversing these two steps. By graph approximation technique, we present an eigenvector decomposition algorithm acting on a $2n \times 2n$ matrix with provable approximation guarantees.

Suppose that $\mathbf{g}$ is an eigenvector of $\mathbf{\overline{P}}$ of $2m$ dimensions on directed edges, then based on the definition of \emph{in-sum} and \emph{out-sum}, $\mathbf{g}^{in}$ and $\mathbf{g}^{out}$ are vectors of $n$ dimensions after performing \emph{in-sum} and \emph{out-sum} operations. If these exists a $2n\times 2n$ matrix $\mathbf{T}$ and a $2m\times 2m$ matrix $\mathbf{Q}$, such that
\begin{equation}
\left(
\begin{aligned}
&(\mathbf{g}^T\mathbf{\overline{P}})^{in} \\
&(\mathbf{g}^T\mathbf{\overline{P}})^{out}
\end{aligned}
\right)
\approx
\left(
\begin{aligned}
&(\mathbf{g}^T\mathbf{Q})^{in} \\
&(\mathbf{g}^T\mathbf{Q})^{out}
\end{aligned}
\right)
=\left(
\mathbf{T}
\left(
\begin{aligned}
&\mathbf{g}^{in} \\
&\mathbf{g}^{out}
\end{aligned}
\right)
\right).
\end{equation}
This implies that if matrix $\mathbf{Q}$ adequately approximates $\mathbf{\overline{P}}$, without operating on matrix $\boldsymbol{\overline{P}}$, one can perform spectral decomposition directly on $\mathbf{T}$, which is much smaller than $\mathbf{\overline{P}}$, to get $\mathbf{g}^{in}$ and $\mathbf{g}^{out}$. We can view $\mathbf{T}$ as an aggregating version of matrix $\mathbf{\overline{P}}$, which means that $\mathbf{T}$ contains almost the same amount of information as $\mathbf{\overline{P}}$ for our embedding purpose. Next, to compose $\mathbf{T}$ matrix, for any node $u$, we consider its \emph{out-sum} operation when applying $\mathbf{\overline{P}}$ matrix on it.

\begin{lemma}
\label{lemma: noise}
There exists a $2m\times 2m$ matrix $\mathbf{Q}=\mathbf{\overline{P}}+\mathbf{\Delta}$, such that for arbitrary node $x$, $u$, $v\in V$, $|\mathbf{\Delta}_{[(x\rightarrow u),(u\rightarrow v)]}|\leq \frac{1}{2(d(u)-1))(d(v)-1)}$ if $\overline{P}_{[(x\rightarrow u),(u\rightarrow v)]}\neq 0$, otherwise 0. Moreover, $(\mathbf{g}^T\mathbf{\overline{P}})^{out}_{u}$  can be approximated as $(\mathbf{g}^T\mathbf{\overline{P}})^{out}_{u}\approx(\mathbf{g}^T\mathbf{Q})^{out}_{u}=(\frac{1}{2}-\sum_{v\in N(u)}\frac{1}{2(d(v)-1)}\frac{1}{d(u)})g^{in}_u+\sum_{v\in N(u)}\frac{1}{2(d(v)-1)}g^{out}_v$.
\end{lemma}
\begin{proof}
The proof is given in the Appendix A.5.
\end{proof}
Likewise, for \emph{in-sum} operation, $(\mathbf{g}^T\mathbf{\overline{P}})^{in}_{u}\approx(\mathbf{g}^T\mathbf{Q})^{in}_{u}=(\frac{1}{2}-\sum_{v\in N(u)}\frac{1}{2(d(v)-1)}\frac{1}{d(u)})g^{out}_u +\sum_{v\in N(u)}\frac{1}{2(d(v)-1)}g^{in}_v $.
After removing constant factors and transforming this formula into a matrix form, we have
\begin{equation}
\mathbf{T}=\left[
\begin{aligned}
&\mathbf{J}&\mathbf{ I}-\mathbf{D}^{-1}\mathbf{J}\\
& \mathbf{I}-\mathbf{D}^{-1}\mathbf{J}&\mathbf{J}
\end{aligned}
\right],
\end{equation}
where $\mathbf{I}$ is the identity matrix and $\mathbf{J}=\mathbf{A}(\mathbf{D}-\mathbf{I})^{-1}$. By switching the order of spectral decomposition and summation, the approximation target is achieved. Now, our graph approximation algorithm NOBE-GA is just directly selecting the second to the $k+1$ largest eigenvectors of $\mathbf{T}$ as our embedding vectors. As these eigenvectors of $2n$ dimensions have \emph{in-sum} and \emph{out-sum} embedding parts, consistent with NOBE, we simply choose \emph{in-sum} part as the final node embeddings.

To prove the approximation guarantee of NOBE-GA, we first introduce some basic notations from spectral graph theory. For a matrix $\mathbf{A}$, we write $\mathbf{A}\succcurlyeq 0$, if $\mathbf{A}$ is positive semi-definite, Similarly, we write $\mathbf{A}\succcurlyeq \mathbf{B}$, if $\mathbf{A}-\mathbf{B}\succcurlyeq 0$, which is also equivalent to $\mathbf{v}^T\mathbf{A}\mathbf{v} \succcurlyeq \mathbf{v}^T\mathbf{B}\mathbf{v}$, for all $\mathbf{v}$. For two graphs $\mathcal{G}$ and $\mathcal{H}$ with the same node set, we denote $\mathcal{G}\succcurlyeq \mathcal{H}$ if their Laplacian matrix $\mathbf{L}_\mathcal{G}\succcurlyeq \mathbf{L}_\mathcal{H}$. Recall that $\mathbf{x}^T\mathbf{L}_\mathcal{G}\mathbf{x}=\sum_{(u,v)\in E}W_{\mathcal{G}}(u,v)(x(u)-x(v))^2$, where $W_{\mathcal{G}}(u,v)$ denotes an item in weighted adjacency matrix of $\mathcal{G}$. It is clear that dropping edges will decrease the value of this quadratic form. Now, we define the approximation between two graphs based on the difference of their Laplacian matrices.
\begin{Definition}
\label{de:AG}
(\textbf{c-approximation graph}) For some $c>1$, a graph $\mathcal{H}$ is called a $c$-approximiation graph of graph $\mathcal{G}$, if $c\mathcal{H}\succcurlyeq \mathcal{G}\succcurlyeq \frac{\mathcal{H}}{c}.$
\end{Definition}
Based on the Definition \ref{de:AG}, we present the Theorem \ref{theorem:AG}, which further shows the relationship of $\mathcal{G}$ and its $c$-approximation graph $\mathcal{H}$ in terms of their eigenvalues.
\vspace{-5pt}
\begin{theorem}
\label{theorem:AG}
If $\mathcal{H}$ is a $c$-approximation graph of graph $\mathcal{G}$, then
$$
|\lambda_k(\mathcal{G})-\lambda_k(\mathcal{H})|\leq max\{(c-1),(1-\frac{1}{c})\}\lambda_k(\mathcal{G}),
$$
where $\lambda_k(\cdot)$ is the k-th smallest eigenvalue of the corresponding graph.
\end{theorem}
\begin{proof}
The proof is given in the Appendix A.4.
\end{proof}
To relax the strict conditions in Definition \ref{de:AG}, we define a probabilistic version of the $c$-approximation graph by using element-wise constraints.
\begin{Definition}
\label{de:delta-AG}
(\textbf{$(c,\eta)$-approximation graph}) For some $c>1$, a graph $\mathcal{H}=(V,\boldsymbol{\cdot},W_{\mathcal{H}})$ is called a $(c,\eta)$-approximiation graph of graph $\mathcal{G}=(V,\boldsymbol{\cdot},W_{\mathcal{G}}))$, if  $cW_\mathcal{H}(u,v)\leq W_\mathcal{G}(u,v)\leq \frac{1}{c}W_\mathcal{H}(u,v)$ is satisfied with probability at least $1-\eta$.
\end{Definition}
A probabilistic version of Theorem \ref{theorem:AG} follows accordingly. At last, we claim that matrix $\mathbf{Q}$ approximates matrix $\mathbf{\overline{P}}$ well by the following Theorem, which means the approximation of NOBE-GA is adequate.
\begin{theorem}
\label{theorem:bound}
Suppose that the degree of the original graph $\mathcal{G}$ obeys Possion distribution with parameter $\lambda$, i.e., $d\sim \pi(\lambda)$. Then, for some small $\delta$, graph $\mathcal{\widetilde{H}}=(\vec{E},\boldsymbol{\cdot},\mathbf{Q})$ is a $(c,\eta)$-approximation graph of the graph $\mathcal{H}=(\vec{E},\boldsymbol{\cdot},\mathbf{\overline{P}})$, where $c$ is $1+\delta$, and $\eta$ is $\frac{1}{\delta}[\frac{1}{\lambda-1}+\frac{\lambda}{(\lambda-1)^3}]$.
\end{theorem}
\begin{proof}
The proof is given in the Appendix A.6.
\end{proof}
\subsection{Time Complexity and Discussion}
A sparse implementation of our algorithm in Matlab is publicly available\footnote{https://github.com/Jafree/NOnBacktrackingEmbedding}. For $k$-dimensional embedding, summation operation of NOBE requires $O(mk)$ time. Eigenvector computation, which utilizes a variant of Lanczos algorithm, requires $O(\widetilde{t}k)$ time, where $\widetilde{t}=\sum_v2d(v)[d(v)-1]$. In total, the time complexity of NOBE is $O(\widetilde{t}k+mk)$. The time complexity of NOBE-GA is $O(nk\overline{d})$, where $\overline{d}$ is the average degree.

Classic spectral method based on $\mathbf{D^{-1}A}$ is just a reduced version of NOBE (proof in Appendix A.7). For sparse and degree-skewed networks, nodes with a large degree will affect many related eigenvectors. Therefore, the previous leading eigenvector that corresponds to community structure will be lost in the bulk of worthless eigenvectors, and hence fail to preserve meaningful structures \cite{krzakala2013spectral}. However, our spectral framework with non-backtracking strategy can overcome this issue.
\section{Experimental Results}
\label{sec:Experiments}
In this section, we first introduce datasets and compared methods used in the experiments. After that, we present empirical evaluations on clustering and structural hole spanner detection in detail.
\subsection{Dataset Description}
All networks used here are undirected, which are publicly available on SNAP dataset platform \cite{snapnets}. They vary widely from a range of characteristics such as network type, network size and community profile. They include three social networks: \emph{karate} (real), \emph{youtube} (online), \emph{enron-email} (communication); three collaboration networks: \emph{ca-hepth}, \emph{dblp}, \emph{ca-condmat} (bipartite of authors and publications); three entity networks: \emph{dolphins} (animals), \emph{us-football} (organizations), \emph{polblogs} (hyperlinks). The summary of the datasets is shown in Table \ref{tab:datasets}. Specifically, we apply a community detection algorithm, i.e., RanCom \cite{jiang2015fast}, to show the detected community number and maximum community size.

\subsection{Compared Methods}
We compare our methods with the state-of-the-art algorithms. The first three are graph embedding methods. The others are SH spanner detection methods . We summarize them as follows:
\begin{itemize}
\setlength{\itemsep}{2pt}
\setlength{\parsep}{2pt}
\setlength{\parskip}{2pt}
\item \textbf{NOBE}, \textbf{NOBE-GA}: Our spectral graph embedding method and its graph approximation version.
\item \textbf{node2vec} \cite{grover2016node2vec}: A feature learning framework extending Skip-gram architecture to networks.

\item \textbf{LINE} \cite{tang2015line}: The version of combining first-order and second-order proximity is used here. 

\item \textbf{Deepwalk} \cite{perozzi2014deepwalk}: Truncated random walk and language modeling techniques are utilized.

\item \textbf{HAM} \cite{he2016joint}: A harmonic modularity function is proposed to tackle the SH spanner detection.

\item \textbf{Constraint} \cite{burt2009structural}: A constraint introduced to prune nodes with certain connectivity being candidates.

\item \textbf{Pagerank} \cite{page1999pagerank}: Nodes with highest pagerank score will be selected as SH spanners.

\item \textbf{Betweenness Centrality} (BC) \cite{brandes2001faster}: Nodes with highest BC will be selected as SH spanners.

\item \textbf{HIS} \cite{lou2013mining}: Designing a two-stage information flow model  to optimize the provided objective function.

\item \textbf{AP\_BICC} \cite{rezvani2015identifying}: Approximate inverse closeness centralities and articulation points are exploited.
\end{itemize}
\begin{table}[tbp]
\scriptsize
\caption{\bf Summary of experimental datasets and their community profiles.}
\centering 
\begin{tabular}{|C{1.55cm}|C{1cm}C{1cm}C{1.4cm}C{1.4cm}|}

\hline
\multirow{2}{*}{}&\multicolumn{2}{c}{\textbf{Characteristics}}&\textbf{\#Community}&\textbf{\#Max members}\\
\cline{2-5}Datasets&\# Node&\# Edge&\emph{RankCom}&\emph{RankCom}\\
\hline
\emph{karate}&34&78&2   & 18  \\

\emph{dolphins}&62&159&3   & 29 \\

\emph{us-football}&115&613&11   & 17  \\

\emph{polblogs}&1,224&19,090&7  & 675  \\

\emph{ca-hepth}&9,877&25,998&995  & 446  \\

\emph{ca-condmat}&23,133&93,497&2,456  & 797  \\

\emph{email-enron}&36,692&183,831&3,888  & 3,914  \\

\emph{youtube}&334,863&925,872&15,863  & 37,255 \\

\emph{dblp}&317,080&1,049,866&25,633  & 1,099 \\

\hline
\end{tabular}
\label{tab:datasets}
\end{table}
\vspace{-10pt}
\subsection{Performance on Clustering}
\begin{table*}[htbp]
\scriptsize
\caption{
\bf{Performance on Clustering evaluated by Modularity and Permanence$_{(rank)}$}}
\centering 
\begin{tabular}{C{1.3cm}|C{1.1cm}||C{0.7cm}C{1.2cm}C{0.9cm}C{0.7cm}C{1.2cm}||C{0.8cm}C{0.9cm}C{1.1cm}C{1.1cm}C{1.1cm}}
\hline
\multirow{2}{*}{}&&\multicolumn{5}{c}{Modularity}&\multicolumn{5}{c}{Permanence}\\

\cline{3-12}Datasets& Clustering Methods&\textbf{NOBE}&\textbf{NOBE-GA}&\textbf{node2vec}&\textbf{LINE}&\textbf{Deepwalk}&\textbf{NOBE}&\textbf{NOBE-GA}&\textbf{node2vec}&\textbf{LINE}&\textbf{Deepwalk} \\

\hline
\multirow{2}{*}{\emph{karate}}&\emph{k-means}& \textbf{0.449}(1)& \textbf{0.449}(1)& 0.335(5)&  0.403(3)&0.396(4)& \textbf{0.350}(1)& \textbf{0.350}(1)& 0.335(4)&  0.182(5)&\textbf{0.350}(1)  \\
&\emph{AM}& \textbf{0.449}(1)& \textbf{0.449}(1)& 0.335(4)&  0.239(5)& 0.430(3)& \textbf{0.356}(1)& 0.350(2)& 0.205(5)&  0.232(4)& 0.311(3)\\
\cline{1-2}

\multirow{2}{*}{\emph{dolphins}}&\emph{k-means}& 0.510(2)& \textbf{0.522}(1)& 0.460(3)&  0.187(5)&0.401(4) & 0.250(2)& \textbf{0.268}(1)& 0.196(3)&  -0.166(5)&0.187(4)\\
&\emph{AM}& 0.514(2)&\textbf{ 0.522}(1)& 0.458(3)&  0.271(5)& 0.393(4)& 0.233(2)&\textbf{0.249}(1)& 0.132(4)&  -0.189(5)& 0.189(3)\\
\cline{1-2}

\multirow{2}{*}{\emph{us-footbal}}&\emph{k-means}& 0.610(2)& \textbf{0.611}(1)& 0.605(3)& 0.562(4) & 0.464(5)& \textbf{0.321}(1)& \textbf{0.321}(1)& 0.304(3)& 0.311(2) & 0.039(5)\\
&\emph{AM}& \textbf{0.612}(1)& 0.609(2)& 0.589(3)& 0.492(4) & 0.464(5)& \textbf{0.330}(1)& \textbf{0.330}(1)& 0.279(4)& 0.307(3) & 0.039(5)\\
\cline{1-2}

\multirow{2}{*}{\emph{ca-hepTh}}&\emph{k-means}& \textbf{0.639}(1)& 0.609(2)& 0.597(3)& 0.01(5)& 0.424(4)& \textbf{0.412}(1)& 0.337(3)& 0.379(2)& -0.948(5) & 0.261(4)\\
&\emph{AM}& \textbf{0.635}(1)& 0.614(2)& 0.606(3)& 0.05(5)& 0.453(4)& \textbf{0.435}(1)& 0.416(2)& 0.406(3)& -0.949(5) & 0.338(4)\\
\cline{1-2}

\multirow{2}{*}{\emph{condmat}}&\emph{k-means}& \textbf{0.515}(1)& 0.495(3)&\textbf{0.515}(1) &  0(5)& 0.357(4)& \textbf{0.330}(1)& 0.288(3)&\textbf{0.330}(1)&  -0.984(5)& 0.197(4)\\
&\emph{AM}& \textbf{0.528}(1)& 0.502(3)& 0.520(2)& 0(5) & 0.370(4)& \textbf{0.391}(1)& 0.327(3)& 0.388(2)& -0.994(5) & 0.249(4)\\
\cline{1-2}

\multirow{2}{*}{\emph{enron-email}}&\emph{k-means}& 0.219(2)& \textbf{0.221}(1)& 0.213(3)&  0(5)& 0.178(4)& 0.096(2)& \textbf{0.153}(1)& 0.080(3)&  -0.985(5)& 0.049(4)\\
&\emph{AM}& 0.215(3)& \textbf{0.220}(1)& 0.218(2)& 0(5) & 0.207(4)& 0.120(3)& \textbf{0.194}(1)& 0.180(2)& -0.996(5) & 0.108(4)\\
\cline{1-2}

\multirow{2}{*}{\emph{polblogs}}&\emph{k-means}& \textbf{0.428}(1)&\textbf{0.428}(1)& 0.357(3)&  0.200(4)& 0.084(5)& \textbf{0.138}(1)&0.136(2)& -0.066(3)&  -0.569(5)& -0.187(4)\\
&\emph{AM}& \textbf{0.428}(1)& 0.427(2)& 0.376(3)&  0.266(4)& 0.065(5)& \textbf{0.138}(1)& 0.132(2)& -0.096(3)&  -0.509(5)& -0.176(4)\\
\cline{1-2}

\hline
\end{tabular}
\label{tab:modularity_community}
\end{table*}
Clustering is an important unsupervised application used for automatically separating data points into clusters. Our graph embedding method is used for embedding nodes of a graph into vectors, on which clustering method can be directly employed. Two evaluation metrics considered are summarized as follows:
\begin{itemize}
\setlength{\itemsep}{3pt}
\setlength{\parsep}{3pt}
\setlength{\parskip}{3pt}

\item \textbf{Modularity} \cite{newman2006modularity}: Modularity is a widely used quantitative metric that measures the likelihood of nodes' community membership under the perturbation of the Null model. Mathematically,
$
Q=\frac{1}{2m}\sum_{vw}[A_{vw}-\frac{d(v)\cdot d(w)}{2m}]\delta (c_v,c_w),
$
where $\delta$ is the indicator function. $c_v$ indicates the community node $v$ belongs to. In practice, we add a penalty if a clearly wrong membership is predicted.

\item \textbf{Permanence} \cite{chakraborty2014permanence}: It is a vertex-based metric, which depends on two factors: internal clustering coefficient and maximum external degree to other communities. The permanence of a node $v$ that belongs to community $c$ is defined as follows:
$
Perm_c(v) = [\frac{I_c(v)}{E^c_{max}(v)}\times \frac{1}{d(v)}]-[1-C^c_{in}(v)],
$
where $I_c(v)$ is the internal degree. $E^c_{max}(v)$ is the maximum degree that node $v$ links to another community. $C^c_{in}(v)$ is the internal clustering coefficient. Generally, positive permanence indicates a good community structure. To penalize apparently wrong community assignment, $Perm_c(v)$ is set to $-1$, if $d(v)<2E^c_{max}(v)$.
\end{itemize}

For the clustering application, we summary the performance of our methods, i.e., NOBE and NOBE-GA, against three state-of-the-art embedding methods on seven datasets in terms of modularity and permanence in Table \ref{tab:modularity_community}. Two types of classic clustering methods are used, i.e., k-means and agglomerative method (AM). From these results, we have the following observations:
\begin{figure}[b]
\begin{center}
\includegraphics[width=\linewidth]{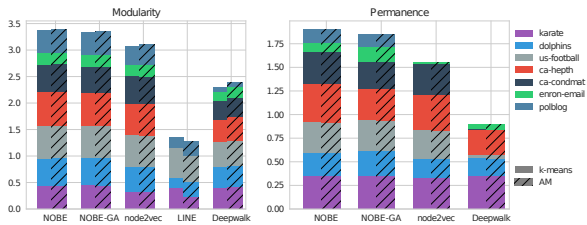}
\end{center}
\caption{
{ The overall performance on clustering in terms of modularity and permanence.}
}
\label{fig:performance_bar}
\end{figure}

\begin{itemize}
\item In terms of modularity and permanence, NOBE and NOBE-GA outperform other graph embedding methods over all datasets under both k-means and AM. Positive permanence scores on all datasets indicate that meaningful community structure is discovered. Specifically, node2vec obtains competing results on \emph{condmat} under k-means. As for LINE, it fails to predict useful community structure on most large datasets except \emph{karate} and \emph{us}-\emph{football}. Deepwalk gives mediocre results on most datasets and bad results on \emph{us}-\emph{football}, \emph{polblogs} and \emph{enron}-\emph{email} under k-means.

\item Figure \ref{fig:performance_bar} reports the overall performance. NOBE and NOBE-GA achieves superior embedding performance for the clustering application. Moreover, it practically demonstrates that NOBE-GA approximates NOBE very well on various kinds of networks despite the difference on their link density, node preferences and community profiles. To our surprise, on some datasets NOBE-GA even achieve slightly better performance than NOBE. We conjecture that this improvement arises because of the introducing of the randomization and the preferences of evaluation metrics. Specifically, in terms of modularity, the percentage of improvement margin of NOBE is $8\%$ over node2vec, $139\%$ over LINE and $39\%$ over Deepwalk. Regarding to permanence, the percentage of the improvement margin of NOBE is $16\%$ over node2vec and $100\%$ over Deepwalk. 

\end{itemize}
\normalsize
\subsection{Performance on Structural Hole Spanner Detection}

\begin{table*}[htbp]
\footnotesize
\caption{
\bf{Structural hole spanner detection results under linear threshold and independent cascade influence models}}
\centering 
\begin{tabular}{cccccccccc}
\hline
\multirow{3}{*}{}& &\multicolumn{7}{c}{Comparative Methods}\\
\cline{4-10}
Datasets&\#SH Spanners&Influence Model&NOBE&HAM&Constraint&PageRank&BC&HIS&AP\_BICC \\
\hline
\multirow{3}{*}{karate}&\multirow{3}{*}{3}&LT&\bf 0.595&0.343&0.295&0.159&0.159&0.132&0.295  \\
& &IC&\bf 0.003&0.002&0.002&0.001&0.001&0.001&0.002  \\
\cline{3-10}
& &SH spanners&[3 20 14]&[3 20 9]&[1 34 3]&[34 1 33]&[1 34 33]&[32 9 14]&[1 3 34]\\
\hline
\multirow{2}{*}{youtube}&\multirow{2}{*}{78}&LT&\bf 4.664&3.951&2.447&1.236&1.226&3.198&1.630  \\
& &IC&\bf 4.375&2.452&1.254&0.662&0.791&2.148&0.799  \\
\cline{3-10}
\multirow{2}{*}{dblp}&\multirow{2}{*}{42}&LT&\bf 8.734&5.384&0.404&0.357&0.958&0.718&0.550  \\
& &IC&\bf 7.221&3.578&0.229&0.190&0.821&0.304&0.495  \\
\hline
\end{tabular}
\label{tab:sh}
\end{table*}

Generally speaking, in a network, structural hole (SH) spanners are the nodes bridging between different communities, which are crucial for many applications such as diffusion controls, viral marketing and brain functional analysis \cite{bassett2009cognitive,lou2013mining,burt2009structural}.  Detecting these bridging nodes is a non-trivial task. To exhibit the power of our embedding method in placing key nodes into accurate positions, we first employ our method to embed the graph into low-dimensional vectors and then detect structural hole spanners in that subspace. We compare our method with SH spanner detection algorithms that are directly applied on graphs. To evaluate the quantitative quality of selected SH spanners, we use a evaluation metric called Structural Hole Influence Index (SHII) proposed in \cite{he2016joint}. This metric is designed by simulating information diffusion processes under certain information diffusion models in the given network. 

\begin{itemize}
\item \textbf{Structural Hole Influence Index} (SHII) \cite{he2016joint}: Regarding a SH spanner candidate $v$, we compute its SHII score by performing the influence maximization process several times.  For each time, to activate the influence diffusion process, we randomly select a set of nodes $S_v$ from the community $C_v$ that $v$ belongs to. Node $v$ and node set $S_v$ is combined as seed set to propagate the influence. After the propagation, SHII score is obtained by computing the relative difference between the number of activated nodes in the community $C_v$ and in other communities: 
$
SHII(v,S_v)=\frac{\sum_{C_i\in \mathcal{C} \backslash C_v} \sum_{u\in C_i}I_u}{\sum_{u\in C_v}I_u},
$
where $\mathcal{C}$ is the set of communities. $I_u$ is the indicator function which equals one if node $u$ is influenced, otherwise $0$.
\end{itemize}

For each SH spanner candidate, we run the information diffusion under linear threshold model (LT) and independent cascade model (IC) 10000 times to get average SHII score. To generate SH spanner candidates from embedded subspace, in which our embedding vectors lie, we devise a metric for ranking nodes:
\begin{itemize}
\item \textbf{Relative Deviation Score} (RDS): Suppose that for each node $v \in V$, its low-dimensional embedding vector is represented as $\mathbf{y}_v \in \mathbb{R}^k$. We apply k-means to separate nodes into appropriate clusters with $\mathcal{C}$ denoting cluster set. For a cluster $C\in \mathcal{C}$, the mean of its points is $\mathbf{u}_C=\frac{1}{|C|}\sum_{i\in C} \mathbf{y}_i$. The Relative Deviation Score, which measures how far a data point is deviating from its own community attracted by other community, is defined as:
$$
\small
RDS(v)=\max_{C\in \mathcal{C}} \frac{\parallel \mathbf{y}_v-\mathbf{u}_{C_v}\parallel_2/R_{C_v}}{\parallel\mathbf{y}_v-\mathbf{u}_{C}\parallel_2/R_C}
$$
\normalsize
where $C_v$ denotes the cluster $v$ belongs to. And $R_C=\sum_{i\in C}\parallel \mathbf{y}_i - \mathbf{u}_C\parallel_2$ indicates the radius of cluster $C$.
\end{itemize}

In our low-dimensional space, nodes with highest RDS will be selected as candidates of SH spanners. We summarize our embedding method against other SH spanner detection algorithms in Table \ref{tab:sh}. Due to space limit, we omit the results of other embedding methods as they totally fail on this task. The number of SH spanners shown in the second column is chosen based on the network size and community profile. Actually, too many SH spanners will lead to the propagation activating the entire network. We outperform all SH spanner detection algorithms under LT and IC models on all three datasets. Specifically, on \emph{karate} network, we identify three SH spanners, i.e., 3, 20 and 14, which can be regarded as a perfect group that can influence both clusters, seen from Figure \ref{fig:karate}. On average, our method NOBE achieves a significant $66\%$ improvement against state-of-the-art algorithm HAM, which shows the power of our method in accurate embedding.

\vspace{-5pt}
\begin{figure}[hbt]
\begin{subfigure}{\linewidth}
\centering
\includegraphics[width=\linewidth]{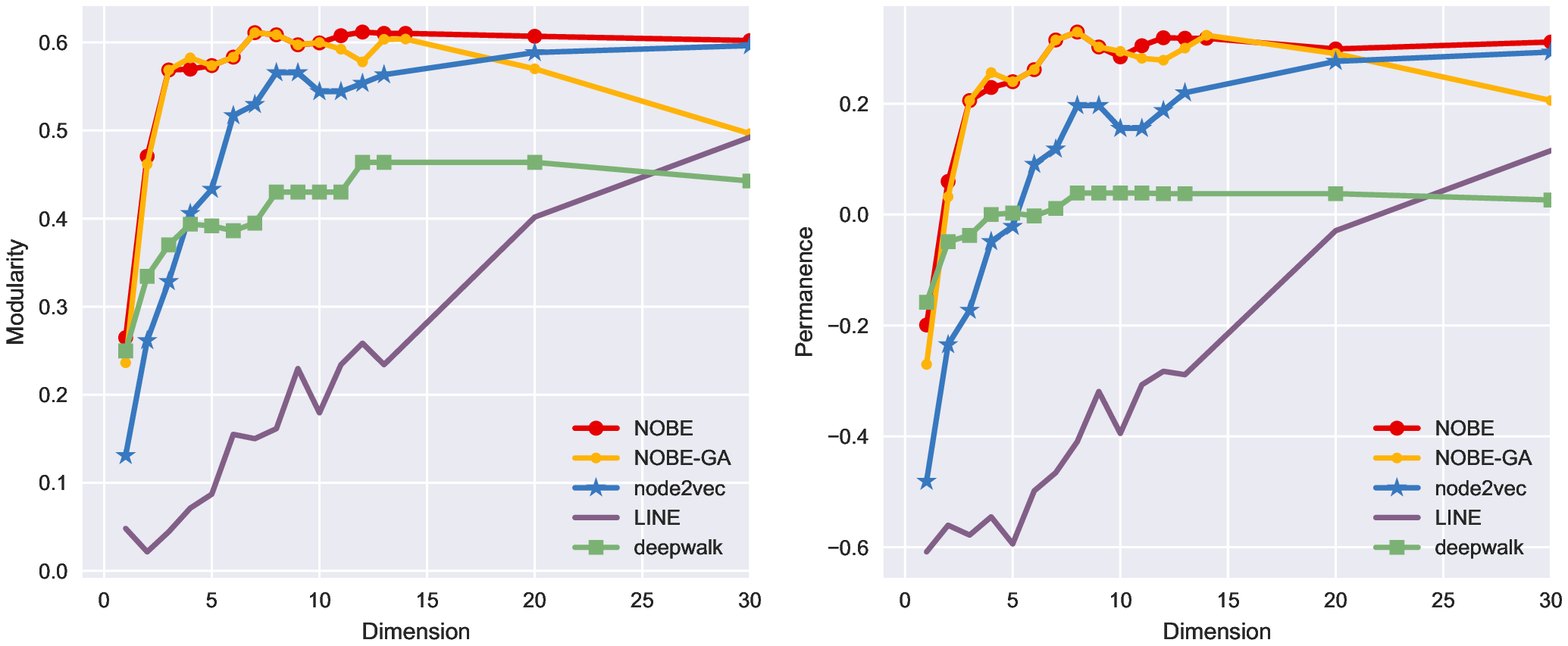}
\vspace{-20pt}
\caption{\emph{us-football}}
\label{fig:football_dimension}
\end{subfigure}
\centering
\begin{subfigure}{\linewidth}
\centering
\includegraphics[width=\linewidth]{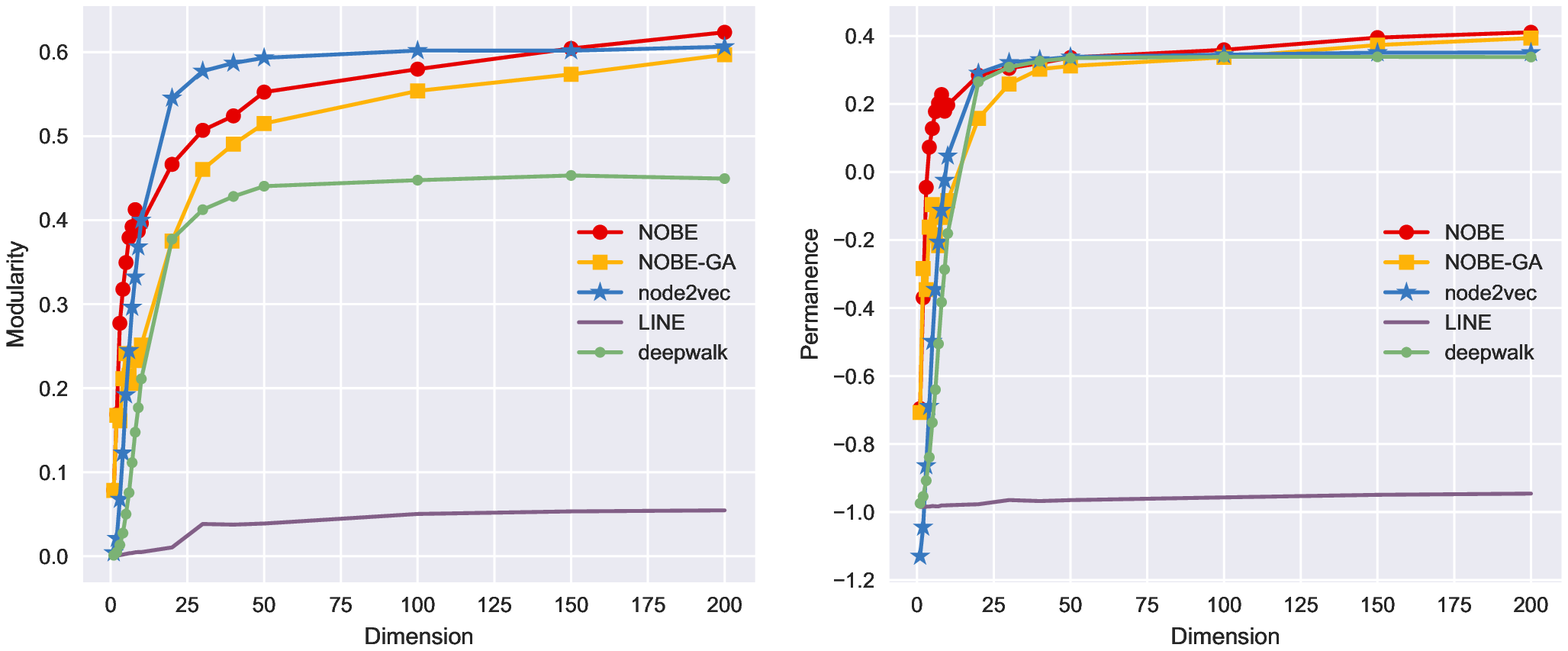}
\vspace{-20pt}
\caption{\emph{ca-hepth}}
\label{fig:hepth_dimension}
\end{subfigure}
\vspace{-5pt}
\caption{
{Parameter Analysis of Dimension under AM.}
}
\label{fig:dimension}
\end{figure}

\subsection{Parameter Analysis}
Dimension is usually considered as a intrinsic characteristic, and often needs to be artificially predefined. With varying dimension, we report the clustering quality under AM on two datasets in Figure \ref{fig:dimension}. On \emph{football} network with 11 ground truth communities, NOBE, NOBE-GA and node2vec achieves reasonable results on dimension $k=7$ or $8$. After $k=11$, NOBE, NOBE-GA, node2vec and deepwalk begin to drop. Followed by a sudden drop, node2vec still increases gradually. Reported by RankCom, \emph{ca}-\emph{hepth} network has 995 communities. Nevertheless, prior to dimension $k=50$, NOBE, NOBE-GA, node2vec and deepwalk have already obtained community structure with good quality. The performance will slightly increase afterwards. Consistent with studies on spectral analysis of graph matrices \cite{shen2010}, community number is a good choice for dimension. However, it's also rather conservative, since good embedding methods could preserve great majority of graph information in much shorter vectors. The choice of a large number greater than community number should be cautious since redundant information added may deteriorate embedding results .

\section{Conclusion and Outlook}
\label{sec:Conclusion}
This paper proposes NOBE, a novel framework leveraging the non-backtracking strategy for graph embedding. It exploits highly nonlinear structure of graphs by considering a non-Markovian dynamics. As a result, it can handle both macroscopic and microscopic tasks. Experiments demonstrate the superior advantage of our algorithm over state-of-the-art baselines. In addition, we carry out a graph approximation technique with theoretical guarantees for reducing the complexity and also for analyzing the different flows on graphs. To our surprise, NOBE-GA achieves excellent performance at the same level as NOBE.

We hope that our work will shed light on the analysis of algorithms based on flow dynamics of graphs, especially on spectral algorithms. Graph approximation can be further investigated by considering the perturbation of eigenvectors. We leave it for future work.

\section*{Acknowledgement}
This work is supported in part by National Key R\& D Program of China through grants 2016YFB0800700, and NSF through grants IIS-1526499, and CNS-1626432, and NSFC 61672313, 61672051, 61503253, and NSF of Guangdong Province 2017A030313339.
 
\bibliographystyle{siamplain}
\bibliography{references}
\appendix
\section{Appendix}
\subsection{}
\textbf{Proposition 3.1} If the minimum degree of the connected graph $\mathcal{G}$ is at least 2, then the oriented line graph $\mathcal{H}$ is valid and strongly connected.
\begin{proof}
Assume that $(w\rightarrow x)$ and $(u\rightarrow v)$ are two arbitrary nodes in the oriented line graph $\mathcal{H}$. The proposition is equivalent to prove that $(u\rightarrow v)$ can be reached from $(w\rightarrow x)$. Three situations should be considered:

\noindent
1) if $x=u$ and $w\neq v$, then $(w\rightarrow x)$ is directly linked to $(u\rightarrow v)$; 

\noindent
2) if $w=v$ and $x\neq u$ which means there is a directed edge from $(u\rightarrow v)$ to $(w\rightarrow x)$. We delete the node $w$, i.e., node $v$, in the original graph $\mathcal{G}$. Since the minimum degree of $\mathcal{G}$ is at least two. Therefore, node $u$ and node $x$ are still mutually reachable in graph $\mathcal{G}$. A Hamilton Path $\mathbf{p}$  from node $x$ to node $u$ can be selected with passing through other existing nodes only once, which satisfies the non-backtracking condition. Adding node $w$, i.e., node $v$, back into the graph $\mathcal{G}$ will generate a non-backtracking path $(w\rightarrow\mathbf{p}\rightarrow v)$, which means $(u\rightarrow v)$ is reachable from $(w\rightarrow x)$ in the oriented line graph $\mathcal{H}$;

\noindent
3) if $w,x,u,v$ are mutually unequal. Assume that we delete edges $(w,x)$ and $(u,v)$ in graph $\mathcal{G}$, then graph $G$  is still connected. There exists a Hamilton path $\mathbf{p}$ connecting node $x$ and $u$. Thus, with $(w\rightarrow\mathbf{p}\rightarrow v)$ satisfying the non-backtracking condition, there exists a directed path connecting node $(w\rightarrow x)$ and node $(u\rightarrow v)$ in the oriented line graph $\mathcal{H}$.

\noindent
Overall, every valid node in graph $\mathcal{H}$ can be reached, if graph $\mathcal{G}$ has a minimum degree at least two.
\end{proof}

\subsection{}
\noindent
\textbf{Proposition 3.2} Our loss function is 
\begin{equation*}
\min_{\mathbf{y}} \mathbf{y}^T(\mathbf{\Phi} - \frac{\mathbf{\Phi} \mathbf{P}+\mathbf{P}^T\mathbf{\Phi}}{2})\mathbf{y}.
\end{equation*}

\begin{proof}
To be concise, we use $V(\mathcal{H})$, $E(\mathcal{H})$ denote the node set and the edge set of $\mathcal{H}$ separately. $u$,$v$ denote nodes in $V(\mathcal{H})$. Considering every node pair, we have the following loss function
\small
\begin{equation*}
\sum\limits_{u,v \in V(\mathcal{H})}\{\phi(u)[(y(u)-y(v))^2 P(u,v)]+\phi(v)[(y(v)-y(u))^2 P(v,u)]\}.
\end{equation*}
\normalsize
Dividing the term into two parts by regarding each node pair as ordered and expanding the formula, we get
\small
\begin{equation}
\label{eq:two-term}
\begin{aligned}
&\frac{1}{2}\sum\limits_{u\in V(\mathcal{H})}\sum\limits_{(u,v)\in E(\mathcal{H})}[\phi(u)(y(u)^2+y(v)^2-2y(u)y(v)) P(u,v)] \\
+&\frac{1}{2}\sum\limits_{u\in V(\mathcal{H})}\sum\limits_{(v,u)\in E(\mathcal{H})}[\phi(v)(y(u)^2+y(v)^2-2y(u)y(v)) P(v,u)]
\end{aligned}
\end{equation}
\normalsize
Then, we do the deduction for the first part. A similar proof can be applied to the second part. Enumerating each node in the first part, we get
\begin{equation}
\label{eq:firstterm}
\begin{aligned}
&\frac{1}{2}\sum\limits_{u\in V(\mathcal{H})}\phi(u)y(u)^2\sum\limits_{v:(u,v)\in E(\mathcal{H})}p(u,v)\\
&+\frac{1}{2}\sum\limits_{v\in V(\mathcal{H})}y(v)^2\sum\limits_{u:(u,v)\in E(\mathcal{H})}(\phi(u) P(u,v))\\
&-\frac{1}{2}\sum\limits_{v:(u,v)\in E(\mathcal{H})}\phi(u)(2y(u)y(v)P(u,v))
\end{aligned}
\end{equation}
Due to proposition 3.3, the sum of each row is equal to one, i.e., $$\sum\limits_{v:(u,v)\in E(\mathcal{H})}P(u,v)=1.$$ The first part of above equation becomes 
\begin{equation*}
\frac{1}{2}\mathbf{y}^T\mathbf{\Phi} \mathbf{y}.
\end{equation*}
Here, $\mathbf{\Phi}$ is a diagonal matrix with $\Phi(i,i)=\phi(i).$
Since $\mathbf{\phi}^T \mathbf{P}=\mathbf{\phi}^T$, so the second sum in the second term of equation \ref{eq:firstterm} becomes $$\sum\limits_{u:(u,v)\in E(\mathcal{H})}(\phi(u) P(u,v))=\phi(v).$$ Then, the matrix form of the second term in equation \ref{eq:firstterm} becomes $$\frac{1}{2}\mathbf{y}^T\mathbf{\Phi} \mathbf{y}$$
Arranging the terms in a particular order, we can easily see the third part in equation \ref{eq:firstterm}
$$
-\frac{1}{2}\sum\limits_{v:(u,v)\in E(\mathcal{H})}2y(u)\phi(u)P(u,v)y(v)=\mathbf{y}^T\mathbf{\Phi} \mathbf{P} \mathbf{y}.
$$
Adding up all terms and removing the constant factor, we get our loss function
$$
\min_{\mathbf{y}} \mathbf{y}^T(\mathbf{\Phi} - \frac{\mathbf{\Phi} \mathbf{P}+\mathbf{P}^T\mathbf{\Phi}}{2})\mathbf{y}.
$$
\end{proof}

\subsection{}
\textbf{Proposition 3.3} Both the sums of rows and columns of the non-backtracking transition matrix $\mathbf{P}$ equal one.
\begin{proof}
The proof is simple, since concerning each row or column, the values of nonzero items are equal. For an arbitrary row related to node $(u\rightarrow v)$, the sum of this row in the non-backtracking matrix $\mathbf{P}$ is 
$$
\begin{aligned}
\sum\limits_{\substack{x\in N(v) \\ x\neq u}} \mathbf{P}_{[(u\rightarrow v),(v\rightarrow x)]}&=\sum\limits_{\substack{x\in N(v) \\ x\neq u}} \frac{1}{d(v)-1}\\
&= \frac{1}{d(v)-1}\sum\limits_{\substack{x\in N(v) \\ x\neq u}}1 =1
\end{aligned}
$$
Similarly, we can get the same result for each column.
\end{proof}

\subsection{}
\textbf{Theorem 3.1} 
If $\mathcal{G}$ and $\mathcal{H}$ are graphs such that $\mathcal{H}$ is a $c$-approximation graph of graph $\mathcal{G}$, then
$$
|\lambda_k(\mathcal{G})-\lambda_k(\mathcal{H})|\leq max\{(c-1),(1-\frac{1}{c})\}\lambda_k(\mathcal{G}),
$$
where $\lambda_k$ is the k-th smallest eigenvalue of corresponding graph.

\begin{proof}
\label{proof:AG}
Applying Courant-Fisher Theorem, we have
$$
\lambda_k(\mathcal{G})=\min\limits_{\substack{S\subseteq R^n\\ dim(S)=k}} \max\limits_{\mathbf{x}\in S}\frac{\mathbf{x}^T\mathbf{L}_\mathcal{G}\mathbf{x}}{\mathbf{x}^T\mathbf{x}}.
$$
As $\mathcal{H}$ is a $c$-approximation graph of $\mathcal{G}$, then $\mathbf{L}_\mathcal{G}\succcurlyeq \frac{1}{c}\cdot \mathbf{L}_\mathcal{H}$. So, we have
$$
\mathbf{x}^T\mathbf{L}_\mathcal{G}\mathbf{x}\geq \frac{1}{c}\mathbf{x}^T\mathbf{L}_\mathcal{H}\mathbf{x}.
$$
Then, it becomes to
\begin{equation}
\begin{aligned}
\lambda_k(\mathcal{G})&=\min\limits_{\substack{S\subseteq R^n\\ dim(S)=k}} \max\limits_{\mathbf{x}\in S}\frac{\mathbf{x}^T\mathbf{L}_\mathcal{G}\mathbf{x}}{\mathbf{x}^T\mathbf{x}}\geq \min\limits_{\substack{S\subseteq R^n\\ dim(S)=k}} \max\limits_{\mathbf{x}\in S}\frac{\frac{1}{c}\mathbf{x}^T\mathbf{L}_\mathcal{H}\mathbf{x}}{\mathbf{x}^T\mathbf{x}}\\
&=\frac{1}{c}\min\limits_{\substack{S\subseteq R^n\\ dim(S)=k}} \max\limits_{\mathbf{x}\in S}\frac{\mathbf{x}^T\mathbf{L}_\mathcal{H}\mathbf{x}}{\mathbf{x}^T\mathbf{x}}=\frac{1}{c}\lambda_k(\mathcal{H}).
\end{aligned}
\end{equation}
Similarly, we can get $c\lambda_k(\mathcal{H})\geq \lambda_k(\mathcal{G})$. In other words, $c\lambda_k(\mathcal{G})\geq \lambda_k(\mathcal{H}) \geq \frac{1}{c}\lambda_k(\mathcal{G})$. Easy math will give the final result.
\end{proof}

\subsection{}
\textbf{Lemma 3.1}
There exists a $2m\times 2m$ matrix $\mathbf{Q}=\mathbf{\overline{P}}+\mathbf{\Delta}$, such that for arbitrary node $x$, $u$, $v\in V$, if $\overline{P}_{[(x\rightarrow u),(u\rightarrow v)]}\neq 0$, then $|\mathbf{\Delta}_{[(x\rightarrow u),(u\rightarrow v)]}|\leq \frac{1}{2(d(u)-1))(d(v)-1)}$ otherwise 0. Moreover, $(\mathbf{g}^T\mathbf{\overline{P}})^{out}_{u}$  can be approximated as $(\mathbf{g}^T\mathbf{\overline{P}})^{out}_{u}\approx(\mathbf{g}^T\mathbf{Q})^{out}_{u}=(\frac{1}{2}-\sum_{v\in N(u)}\frac{1}{2(d(v)-1)}\frac{1}{d(u)})g^{in}_u+\sum_{v\in N(u)}\frac{1}{2(d(v)-1)}g^{out}_v$.

\begin{proof}
\normalsize
Considering the vector $\mathbf{g}$ as the information contained on each directed edge in graph $\mathcal{G}$, from the definition of \emph{out} operation and the non-backtracking transition matrix, $(\mathbf{g}^T\mathbf{\overline{P}})^{out}_{u}$ is equivalent to a process that first applying one step random walk with probability transition matrix $\mathbf{\overline{P}}$ to update the vector $\mathbf{g}$, and then conducting $out$ operation on node $u$ in graph $\mathcal{G}$. So, we can get
\begin{equation}
\footnotesize
\label{eq:gout}
\begin{aligned}
&(\mathbf{g}^T\mathbf{\overline{P}})^{out}_{u}
=\sum\limits_{v\in N(u)}(\mathbf{g}^T\mathbf{\overline{P}})_{u\rightarrow v}\\
&=\sum\limits_{v\in N(u)}(\sum\limits_{\substack{ x\in N(u) \\ x\neq v}}\frac{1}{2(d(u)-1)}g_{x\rightarrow u}+\sum\limits_{\substack{y\in N(u) \\ y\neq u}}\frac{1}{2(d(v)-1)}g_{v\rightarrow u})
\end{aligned}
\end{equation}
\normalsize
For the first part in the second equation of equation \ref{eq:gout}, by separating the non-backtracking part and switching the summation we have
\begin{equation}
\small
\label{eq:gout1}
\begin{aligned}
&\sum\limits_{v\in N(u)}\sum\limits_{\substack{ x\in N(u) \\ x\neq v}}\frac{1}{2(d(u)-1)}g_{x\rightarrow u}\\
&=\sum\limits_{v\in N(u)}(\sum\limits_{ x\in N(u)}\frac{1}{2(d(u)-1)}g_{x\rightarrow u}-\frac{1}{2(d(u)-1}g_{v\rightarrow u})\\
&=(\sum\limits_{v\in N(u)}\frac{1}{2(d(u)-1)})\sum\limits_{ x\in N(u)}g_{x\rightarrow u}-\frac{1}{2(d(u)-1}\sum\limits_{v\in N(u)}g_{v\rightarrow u}\\
&=(\frac{d(u)}{2(d(u)-1)}-\frac{1}{2(d(u)-1)})g_{u}^{in}\\
&=\frac{1}{2}g^{in}_u
\end{aligned}
\end{equation}
For the second part of equation \ref{eq:gout}, we have
\begin{equation}
\label{eq:gout2}
\small
\begin{aligned}
&\sum\limits_{v\in N(u)}\sum\limits_{\substack{y\in N(u) \\ y\neq u}}\frac{1}{2(d(v)-1)}g_{v\rightarrow u}\\
&=\sum\limits_{v\in N(u)}\frac{1}{2(d(v)-1)}\sum\limits_{\substack{y\in N(v) \\ y\neq u}}g_{v\rightarrow y}\\
&=\sum\limits_{v\in N(u)}\frac{1}{2(d(v)-1)}(g^{out}_v -g_{v\rightarrow u})\\
&\approx\sum\limits_{v\in N(u)}\frac{1}{2(d(v)-1)}g^{out}_v -\sum\limits_{v\in N(u)}\frac{1}{2(d(v)-1)}\frac{1}{d(u)}g^{in}_u 
\end{aligned}
\end{equation}
The approximation in the third step adopts an idea from mean field theory that we assume that every incoming edge to a fixed node has the same amount of probability. Thus, to give an unbiased estimation, we use the mean of other edges coming into node $u$, i.e., $\frac{1}{d(u)-1}(g^{in}_u-g_{v\rightarrow u})$, to approximate the $g_{v\rightarrow u}$ when going through $(v\rightarrow u)$ is prohibited by non-backtracking strategy. Note that for a neighbor $x\in N(v)$, $\mathbf{\overline{P}}_{[(x\rightarrow v, v\rightarrow u)]}\leq \frac{1}{2(d(v)-1)}$. By above approximation, matrix $\mathbf{Q}$ is obtained with $\mathbf{Q}=\mathbf{\overline{P}}+\mathbf{\Delta}$, and the bound of the approximation error is $|\mathbf{\Delta}_{[(x\rightarrow v, v\rightarrow u)]}|\leq \frac{1}{2(d(v)-1)(d(u)-1)}$. To sum up equation \ref{eq:gout1} and \ref{eq:gout2}, we have
$$
\begin{aligned}
(\mathbf{g}^T\mathbf{\overline{P}})^{out}_{u}&\approx(\mathbf{g}^T\mathbf{Q})^{out}_{u}=(\frac{1}{2}-\sum\limits_{v\in N(u)}\frac{1}{2(d(v)-1)}\frac{1}{d(u)})g^{in}_u \\
&+\sum\limits_{v\in N(u)}\frac{1}{2(d(v)-1)}g^{out}_v
\end{aligned}
$$
The lemma holds.
\end{proof}

\subsection{}
\textbf{Theorem 3.2}
Suppose that the degree of the original graph $\mathcal{G}$ obeys Possion distribution with parameter $\lambda$, i.e., $d\sim \pi(\lambda)$. Then, for arbitrary small $\delta$, graph $\mathcal{\widetilde{H}}=(\vec{E},\boldsymbol{\cdot},\mathbf{Q})$ is a $(c,\eta)$-approximation graph of the graph $\mathcal{H}=(\vec{E},\boldsymbol{\cdot},\mathbf{\overline{P}})$, where $c$ is $1+\delta$, and $\eta$ is $\frac{1}{\delta}[\frac{1}{\lambda-1}+\frac{\lambda}{(\lambda-1)^3}]$.

\begin{proof}
To investigate the relationship between $\mathcal{\widetilde{H}}$ and $\mathcal{H}$, we first consider the relative difference between $\mathbf{\overline{P}}$ and $\mathbf{Q}$. According to lemma 3.1, for an arbitrary nonzero item $\mathbf{\overline{P}}_{[(x\rightarrow u),(u\rightarrow v)]}$ in $\mathbf{\overline{P}}$, we have
$$\mathbf{\overline{P}}_{[(x\rightarrow u),(u\rightarrow v)]}=\frac{1}{2(d(u)-1)}.$$
The maximum value of the corresponding item in $\mathbf{\Delta}$ is 
$$|\mathbf{\Delta}_{[(x\rightarrow u),(u\rightarrow v)]}|=\frac{1}{2(d(v)-1)(d(u)-1)}.$$ 
So, the relative difference between $\mathbf{\overline{P}}$ and $\mathbf{Q}$ is  
$$\frac{|\mathbf{\Delta}_{[(x\rightarrow u),(u\rightarrow v)]}|}{\mathbf{\overline{P}}_{[(x\rightarrow u),(u\rightarrow v)]}} = \frac{1}{d(v)-1}.$$
Due to the arbitrary choice of node $u$ and $v$, we regard the values concerning $u$ and $v$ as random variables. Thus, after applying Markov inequality, we get $$Pr[\frac{1}{d(v)-1}\geq\delta]\leq \frac{1}{\delta}E[\frac{1}{d(v)-1}].$$
Note that due to the convexity of the reciprocal function, applying Jensen's inequality only gives us the lower bound. To get an upper bound of $E[\frac{1}{d(v)-1}]$, we set random variable $X = d(v)-1$ and $Y=\frac{1}{X}$, one can use the Taylor series expansion around $E[X]$:
\begin{equation}
\begin{aligned}
\small
E[Y]=E[\frac{1}{X}]\leq &E\{\frac{1}{E[X]}-\frac{1}{E^2[X]}(X-E[X])\\
&+\frac{1}{E^3[X]}(X-E[X])^2\}\\
=&\frac{E^2[X]+Var[X]}{E^3[X]}\\
=&\frac{(\lambda-1)^2+\lambda}{(\lambda-1)^3} \text{ \emph{(Due to Possion distribution)}}\\
=&\frac{1}{\lambda-1}+\frac{\lambda}{(\lambda-1)^3}
\end{aligned}
\end{equation}
So the upper bound that the relative difference ratio between corresponding elements in $\mathbf{\overline{P}}$ and $\mathbf{Q}$, which is caused by the approximation, is as follows:
$$Pr[\frac{1}{d(v)-1}\geq\delta]\leq \frac{1}{\delta}(\frac{1}{\lambda-1}+\frac{\lambda}{(\lambda-1)^3}).$$
Set $\eta=\frac{1}{\delta}[\frac{1}{\lambda-1}+\frac{\lambda}{(\lambda-1)^3}]$. Then, graph $\mathcal{\widetilde{H}}$ is $(1+\delta,\eta)$-approximation graph of graph $\mathcal{H}$.
\end{proof}

\subsection{}
\textbf{Claim 1} Spectral method based on lazy random walk is a reduced version of our proposed algorithm NOBE.
\begin{proof}

\normalsize
The proof is similar to Lemma 3.1. The detail of the approximation strategy used here is different. Again, regarding the vector $\mathbf{g}$ as the information contained on each directed edge in graph $\mathcal{G}$, $(\mathbf{g}^T\mathbf{\overline{P}})^{out}_{u}$ is equivalent to a process that first applying one step random walk with probability transition matrix $\mathbf{\overline{P}}$ to update the vector $\mathbf{g}$, and then conducting $out$ operation on node $u$ in graph $\mathcal{G}$. So, we can get

\footnotesize
\begin{equation}
\label{eq:goutlazy}
\begin{aligned}
&(\mathbf{g}^T\mathbf{\overline{P}})^{out}_{u}
=\sum\limits_{v\in N(u)}(\mathbf{g}^T\mathbf{\overline{P}})_{u\rightarrow v}\\
&=\sum\limits_{v\in N(u)}(\sum\limits_{\substack{ x\in N(u) \\ x\neq v}}\frac{1}{2(d(u)-1)}g_{x\rightarrow u}+\sum\limits_{\substack{y\in N(u) \\ y\neq u}}\frac{1}{2(d(v)-1)}g_{v\rightarrow u})
\end{aligned}
\end{equation}
\normalsize
For the first part in the second equation of equation \ref{eq:goutlazy}, by separating non-backtracking part and switching the summation we have
\begin{equation}
\label{eq:goutlazy1}
\small
\begin{aligned}
&\sum\limits_{v\in N(u)}\sum\limits_{\substack{ x\in N(u) \\ x\neq v}}\frac{1}{2(d(u)-1)}g_{x\rightarrow u}\\
&=(\sum\limits_{v\in N(u)}\frac{1}{2(d(u)-1)})\sum\limits_{ x\in N(u)}g_{x\rightarrow u}-\frac{1}{2(d(u)-1}\sum\limits_{v\in N(u)}g_{v\rightarrow u}\\
&=\frac{1}{2}g^{in}_u
\end{aligned}
\end{equation}
For the second part of equation \ref{eq:goutlazy}, we have
\begin{equation}
\label{eq:goulazy2}
\small
\begin{aligned}
&\sum\limits_{v\in N(u)}\sum\limits_{\substack{y\in N(u) \\ y\neq u}}\frac{1}{2(d(v)-1)}g_{v\rightarrow u}\\
&=\sum\limits_{v\in N(u)}\frac{1}{2(d(v)-1)}\sum\limits_{\substack{y\in N(v) \\ y\neq u}}g_{v\rightarrow y}\\
&=\sum\limits_{v\in N(u)}\frac{1}{2(d(v)-1)}(g^{out}_v -g_{v\rightarrow u})\\
&\approx\sum\limits_{v\in N(u)}\frac{1}{2(d(v)-1)}g^{out}_v -\sum\limits_{v\in N(u)}\frac{1}{2(d(v)-1)}\frac{1}{d(v)}g^{out}_v\\
&= \sum\limits_{v\in N(u)}\frac{d(v)-1}{2(d(v)-1)d(v)}g^{out}_v =\sum\limits_{v\in N(u)}\frac{1}{2d(v)}g^{out}_v 
\end{aligned}
\end{equation}
The approximation happened on the third step of the above equation by using $\frac{1}{d(v)}g^{out}_v$ to approximate $g_{v\rightarrow u}$. This approximation is straightforward, since it allows one to diffuse information without inspecting the information collecting from which node. To sum up equation \ref{eq:goutlazy1} and equation \ref{eq:goulazy2}, we have
\begin{equation}
\label{eq: lazylastout}
\begin{aligned}
(\mathbf{g}^T\mathbf{\overline{P}})^{out}_{u}\approx \frac{1}{2}g^{in}_u+\sum\limits_{v\in N(u)}\frac{1}{2d(v)}g^{out}_v 
\end{aligned}
\end{equation}

\noindent
Likewise, we can get 
\begin{equation}
\label{eq: lazylastin}
\begin{aligned}
(\mathbf{g}^T\mathbf{\overline{P}})^{in}_{u}\approx\sum\limits_{v\in N(u)}\frac{1}{2d(v)}g^{in}_v + \frac{1}{2}g^{out}_u
\end{aligned}
\end{equation}
Assume that $\mathbf{R}$ is a $2n$ by $2n$ matrix, such that
\begin{equation}
\label{eq:lazyR}
\left(
\begin{aligned}
&(\mathbf{g}^T\mathbf{\overline{P}})^{in} \\
&(\mathbf{g}^T\mathbf{\overline{P}})^{out}
\end{aligned}
\right)
=\left(
\mathbf{R}
\left(
\begin{aligned}
&\mathbf{g}^{in} \\
&\mathbf{g}^{out}
\end{aligned}
\right)
\right)
\end{equation}
This implies that if this equation holds, under the approximation described above, the spectral method is reduced from a spectral decomposition on the $2m\times 2m$ matrix $\mathbf{\overline{P}}$ to that on the $2n \times 2n$ matrix $\mathbf{R}$ by dropping and approximating some information. According to equation \ref{eq:lazyR}, we organize equation \ref{eq: lazylastout} and \ref{eq:lazyR} into a matrix form, we get 
\begin{equation}
\label{eq:Rexplicit}
\mathbf{R}=\left[
\begin{aligned}
&\mathbf{D}^{-1}\mathbf{A}&\mathbf{I}\\
&\mathbf{I}&\mathbf{D}^{-1}\mathbf{A}
\end{aligned}
\right],
\end{equation}
where $\mathbf{I}$ is the identity matrix. Literally, eigenvector decomposition on $\mathbf{R}$ and then selecting values corresponding to \emph{in} component is equivalent to a conventional spectral method on lazy random walk. We set $\mathbf{y}^T=(\mathbf{y}_{in}^T, \mathbf{y}_{out}^T)$ as an eigenvector of matrix $\mathbf{R}$ with eigenvalue $\lambda$. Then, it becomes
\begin{equation}
\label{eq:equation_set}
\left\{
\begin{aligned}
&\mathbf{y}_{out}+\mathbf{D}^{-1}\mathbf{A}\mathbf{y}_{in}=\lambda\mathbf{y}_{in}\\
&\mathbf{y}_{in}+\mathbf{D}^{-1}\mathbf{A}\mathbf{y}_{out}=\lambda\mathbf{y}_{out}
\end{aligned}
\right.
\end{equation}

\noindent
Consider two situations: 

\noindent
1)If $\mathbf{y}_{in}=\mathbf{y}_{out}$, substituting $\mathbf{y}_{out}$ using $\mathbf{y}_{in}$ into the first equation of the equation \ref{eq:equation_set}, we obtain $$(\mathbf{D}^{-1}\mathbf{A}+\mathbf{I})\mathbf{y}_{in}=\lambda\mathbf{y}_{in},$$
where $\mathbf{B}=\frac{1}{2}\mathbf{D}^{-1}\mathbf{A}+\frac{1}{2}\mathbf{I}$ is the transition probability matrix of lazy random walk, which has the same eigenvectors with the transition probability matrix on classic random walk. Therefore, the eigenvectors of the transition probability matrix of the lazy random walk are contained in the \emph{in} part of the eigenvectors of matrix $\mathbf{R}$.

\noindent
2)If $\mathbf{y}_{in}\neq \mathbf{y}_{out}$, substituting $\mathbf{y}_{out}$ from the first equation of equation \ref{eq:equation_set} into the second equation, we get
\begin{equation}
\label{eq:zeta_similar}
(\lambda^2\mathbf{I}+(\mathbf{D}^{-1}\mathbf{A})^2-2\lambda \mathbf{D}^{-1}\mathbf{A}-\mathbf{I})\mathbf{y}_{in}=\mathbf{0}.
\end{equation}
The solution of equation \ref{eq:zeta_similar} is the root of 
\begin{equation}
\label{eq:zeta_similar_det}
det[\lambda^2\mathbf{I}+(\mathbf{D}^{-1}\mathbf{A})^2-2\lambda \mathbf{D}^{-1}\mathbf{A}-\mathbf{I}]=\mathbf{0},
\end{equation}
where $det$ is short for determinant. Equation \ref{eq:zeta_similar_det} has the similar form as the well-known graph zeta function 
$$det[\lambda^2\mathbf{I}+\mathbf{D}-\lambda \mathbf{A}-\mathbf{I}]=\mathbf{0},$$
which facilitates many crucial problems on graphs.

We argue that our framework can offer great opportunities to explore flow dynamics on graphs. This proof shows the flexibility of our framework NOBE and the applicability of graph approximation technique. Spectral decomposition algorithm based on lazy random walk is a reduced version of our proposed algorithm NOBE. The claim holds. 

\end{proof}

\end{document}